\documentclass[]{interact}

\usepackage{epstopdf}
\usepackage{subfigure}
\usepackage{slashbox}
\usepackage{graphicx}
\usepackage{algorithm}
\usepackage{algpseudocode}
\usepackage{spverbatim}
\usepackage{amssymb}
\usepackage{array}
\newcolumntype{P}[1]{>{\centering\arraybackslash}p{#1}}
\usepackage{float}
\usepackage{multicol, multirow}
\usepackage{amsmath, amsthm}
\usepackage{commath}
\usepackage{enumerate}
\usepackage{color}
\usepackage{bm}
\usepackage{hyperref}
\usepackage{url}
\usepackage{mathptmx}      
\usepackage{latexsym}

\usepackage[authoryear,round]{natbib}

\theoremstyle{plain}
\newtheorem{theorem}{Theorem}[section]

\newtheorem{corollary}[theorem]{Corollary}
\newtheorem{proposition}[theorem]{Proposition}
\newtheorem{result}{Result}
\theoremstyle{definition}
\newtheorem{definition}[theorem]{Definition}

\theoremstyle{remark}
\newtheorem{remark}{Remark}

\begin{document}

\title{Estimation and Autoregressive Modeling of Mittag-Leffler Distributions with Application in Finance}

\author{
\name{Monika S. Dhull \textsuperscript{a}\thanks{ Corresponding author\textsuperscript{a}. Email: 1213monikadhull@gmail.com}, Arun Kumar \textsuperscript{b}}
\affil{\textsuperscript{a}Department of Mathematics, Indian Institute of Technology Roorkee, Uttarakhand-247667, India; \textsuperscript{b}Department of Mathematics, Indian Institute of Technology Ropar, Punjab-140001, India}}

\maketitle

\begin{abstract}
Mittag-Leffler distribution is one of the heavy-tailed distributions with flexibility in capturing power-law behaviors. The study of Mittag-Leffler and its variants has become important as it can very well model the complex real-world phenomena including the high frequency financial data. This work primarily deals with the estimation of parameters of Mittag-Leffler (ML($\alpha, \sigma$)) distribution. We estimate the parameters of ML($\alpha, \sigma$) using empirical Laplace transform method. The simulation study indicates that the proposed method provides satisfactory results. The real life application of ML($\alpha, \sigma$) distribution on high frequency trading data is also demonstrated. We also provide the estimation of three-parameter Mittag-Leffler distribution using empirical Laplace transform. Additionally, we establish an autoregressive model of order $1$, incorporating the Mittag-Leffler distribution as marginals in one scenario and as innovation terms in another. We apply empirical Laplace transform method to estimate the model parameters and provide the simulation study for the same.
\end{abstract}


\section{Introduction}\label{sec:ML_intro}
The Mittag-Leffler (ML($\alpha, \sigma$)) distribution which was originally studied by \cite{Pillai}, belongs to the family of heavy-tailed distribution and has gained popularity in recent years. The applications of ML($\alpha, \sigma$) in various fields, such as studying dynamical phenomena in physics \cite{Weron}, modeling river flow data using time series \cite{Jaya2003}, and analyzing financial data \cite{Kozu1994}, has attracted researchers from different domains. Various stochastic processes with ML($\alpha, \sigma$) distribution for e.g., Neveu branching process, renewal process with rare events are also studied in literature \cite{Huillet2016}. \cite{Kerss} studied fractional skellam process which is obtained by taking the difference of two independent fractional Poisson processes is shown to have applications in the modelling of high frequency data. By considering Mittag-Leffler distribution as a waiting time instead of exponential distribution in classical Poisson process, the time-fractional Poisson process are defined in literature. Further, this distribution is infinitely divisible and hence Mittag-Leffler L\'evy process and more general Prabhakar L\'evy processes are defined \cite{Prabh1}. The distribution is ubiquitous in literature, therefore, the estimation of the Mittag-Leffler distribution holds significant importance across various scientific disciplines. However, the density function of ML($\alpha, \sigma$) distribution is in a complex form (infinite series form or integral form) which is the main obstacle in estimation task. Also, the distribution does not have integer order moments. The parameter estimation based on fractional moments for ML($\alpha, \sigma$) was first addressed by \cite{Kozu}. The main issue with the proposed method is that it requires prior selection and information about the true parameters. The method also gives the estimate of $\alpha$ more than $1$, where as the true parameter $\alpha \in (0,1)$. Later, \cite{Cahoy2013} proposed the algorithm based on log transform of ML($\alpha, \sigma$) distribution which is less restrictive and computationally simple. Mittag-Leffler distribution also plays an important role in defining fractional Poisson process, which efficiently captures the main characteristics of the oil futures high frequency trading data \cite{Arun2020}. \par
Following the two-parameter Mittag-Leffler function ML($\alpha, \sigma$) and distribution, its generalization to the three-parameter form, also known as Prabhakar function, gained widespread popularity. The Prabhakar function was employed to introduce a discrete probability distribution \cite{Tomovski2016}. Further, a more general fractional Poisson process based on the Prabhakar function was also defined. \cite{Ovidio2018} used the Prabhakar operator to get the stochastic solution to a generalized fractional partial differential equation. The recent work by \cite{Prabh1} generalized the Mittag-Leffler Levy process and called it as Prabhakar process by extending its Levy measure through Prabhakar function. For more results and applications of Mittag-Leffler function see \cite{Guisti2020}. 

Since only a few studies have been done in the literature to estimate the parameters of the Mittag-Leffler distribution, it is interesting to further explore other methods of estimation. We propose to use empirical Laplace method to estimate the parameters of ML($\alpha,\sigma$) and ML($\alpha,\sigma, \gamma$). The empirical Laplace transform is one of the widely used tools for estimation. \cite{feigin1983} extended the Laplace transform estimation to weighted area technique to estimate the parameters of multi-stage models. Further, in \cite{Gawronski}, the alternative method based on empirical Laplace transform was introduced to estimate the parameters of distributions with regularly varying  tails and also established the results on the consistency of the estimators. \cite{Csorgo1990} used the empirical version of Laplace transform to derive the estimators of some of the compound distributions. Weighted integral involving empirical Laplace transform is also used in the goodness-of-fit test of Rayleigh distribution \cite{Meintanis2003}. We see that the empirical Laplace transform is reliable technique for the estimation purpose.

In the second part of the article, we define AR($1$) model with the ML($\alpha, 1$) distribution as marginals and later we assume that innovation terms are from ML($\alpha,1$) distribution. The study of AR($1$) model using ML($\alpha, 1$) will be interesting as it is heavy-tailed distribution which can easily incorporate the extreme or sudden events present in data. We compute the Laplace transform of the innovation terms for AR($1$) model and use the empirical Laplace transform method to estimate the model parameters $\alpha$ and $\rho$. Many AR($1$) models have been defined with non-Gaussian innovation terms to deal with different times series data for e.g., refer article \cite{Grunwald1996}. It is important to highlight that AR($1$) model with the non-Gaussian distributions are very well explored in literature as it is one of the most simple models in time series analysis. \cite{Sim1990} studied an AR($1$) model employing a Gamma process as the innovation term. For AR models with residuals that conform to a Student's t-distribution, one can refer to publications such as \cite{Christmas2011, Nduka2018}, and the associated references. \cite{Nduka2018} also successfully applied the EM algorithm on AR model with t-distributed innovations.  \par
Rest of the paper is organised as follows. In Section \ref{sec:ML1}, we discuss some of the properties of ML($\alpha, \sigma$) distribution. Subsection \ref{sec:ML2} provides the empirical Laplace transform method for the estimation of ML($\alpha, \sigma$) distribution along with assessment using simulations. Further, we also present a real life application of ML($\alpha, \sigma$) on oil futures high frequency trading data. In Section \ref{sec:ML4}, we extend the use of empirical Laplace transform to estimate the parameters of three-parameter Mittag-Leffler distribution, also known as Prabhakar distribution. Section \ref{sec:ML3} deals with the autoregressive model of order $1$ with ML($\alpha,1$) marginals. We also discuss the case with innovation terms from ML($\alpha,1$) using Laplace transform. Further we also performed simulation study for AR($1$) model for two different values of parameters $\alpha$ and $\rho$.
Section \ref{sec:ML5} concludes the work.

\section{Mittag-Leffler ($\alpha, \sigma$) distribution}\label{sec:ML1}
The ML$(\alpha,\sigma)$ distribution is widely known as the generalization of exponential distribution and reduces to it for $\alpha=1$. \\

\begin{definition}\label{def:ML}
	The random variable $M$ on ($0,\infty$) is said to be distributed as ML$(\alpha, \sigma)$ if the Laplace transform of $M$ has the form $\phi_M(s) = \dfrac{1}{1+\sigma^{\alpha} s^{\alpha}}$, where $0<\alpha \leq 1$ is known as index of stability, $\sigma$ as scale parameter and $s>0$ \cite{Kozu}. 		
\end{definition}
\vspace{0.4cm}
\noindent In particular for $\alpha=1$ in Def. \ref{def:ML}, we obtain the Laplace transform of exponential distribution with mean $\sigma$. The ML$(\alpha, \sigma)$ is also related to geometric stable distribution as ML$(\alpha, \sigma) = $ GS $(\sigma[\cos{\pi\alpha/2}]^{1/\alpha}, 1, 0)$ for $0<\alpha<1.$ Also, note that, $\sigma M \sim$ ML$(\alpha, \sigma)$ if $M \sim$ ML$(\alpha, 1).$ In this section, we initially focus on the one-parameter ML$(\alpha, 1)$ distribution, followed by an investigation into the two-parameter ML$(\alpha, \sigma)$ distribution.
The notable mixture representation of ML$(\alpha, \sigma)$ as defined in \cite{Kozu} is, 
\begin{align}
	M = \sigma X S^{1/\alpha}, \label{mix_rep}
\end{align} where $X$ is standard exponential and $S$ has the distribution function as, $$F_S(x) = 1 + \dfrac{1}{\pi\alpha}\Bigl[\arctan\Bigl(\dfrac{x}{\sin(\pi\alpha)} + \cot(\pi\alpha)\Bigr) -\dfrac{\pi}{2}\Bigr].$$ For $0<\alpha<1$, \cite{Kozu} proposed the integral representations of the density function and distribution function of ML$(\alpha,1)$ as,
\begin{align}
	f_M(x) &= \dfrac{\sin \pi\alpha}{\pi} \int_{0}^{\infty} \dfrac{y^{\alpha}\exp{(-xy)}}{y^{2\alpha}+1+2y^{\alpha}\cos{\pi\alpha}} \,dy, \,\,\, x>0,\\
	F_M(x) &= 1 - \dfrac{\sin \pi\alpha}{\pi}\int_{0}^{\infty} \dfrac{y^{\alpha-1}\exp{(-xy)}}{y^{2\alpha}+1+2y^{\alpha}\cos{\pi\alpha}} \,dy, \,\,\, x>0,
\end{align}
respectively. For more details on the properties of ML$(\alpha, 1)$, refer the review article by \cite{suresh2003}.
\subsection{Estimation of ML($\alpha, \sigma$) distribution}\label{sec:ML2}
In this subsection, we employ the empirical Laplace transform method to estimate the $\alpha$ and $\sigma$ parameters of ML($\alpha, \sigma$).
The Laplace transform for ML($\alpha, \sigma$) with distribution function as $F_M(x)$ is \cite{Kozu} $$\phi_M(s;\alpha,\sigma) = \mathbb{E}[e^{-sM}] = \int_0^{\infty}e^{-sM}\, dF(x) = \dfrac{1}{1+ (\sigma s)^{\alpha}}$$ and the empirical Laplace transform for iid samples $\{X_i\}, i=1,2,\hdots n$ as defined in \cite{Csorgo1990} is, $$\phi_n(s) = \frac{1}{n}\sum_{i=1}^{n}\exp(-s X_i) = \int_0^{\infty} e^{-sX}\, dF_n(x), \text{ where } dF_n(x) = \frac{1}{n}\#\{1 \leq i\leq n, X_i\leq x\}.$$ 
The objective function to employ empirical Laplace transform is 
\begin{align}
Q(\alpha, \sigma) =\displaystyle\sum_{j=1}^{m}[\phi_n(s_j) - \phi_M(s_j; \alpha,\sigma)]^2 ,
\end{align}
where $m=100$ and $0.01 \leq s_j \leq 1$ for estimation. The parameters $\alpha$ and $\sigma$ were estimated by minimizing the empirical Laplace transform criterion using the Nelder-Mead simplex algorithm present in the python scipy package \cite{Gao2012}. Although the numerical optimization was performed without explicitly imposing parameter bounds, the resulting estimates consistently satisfied the theoretical parameter restrictions $0<\widehat{\alpha}\leq1$ and $\widehat{\sigma}>0$ for all simulated data considered. Hence, all reported estimates belonged to the admissible parameter space of the model.
\[
(\widehat{\alpha},\widehat{\sigma})
=
\underset{0<\alpha\leq 1,\;\sigma>0}{\operatorname{arg\,min}}
\sum_{j=1}^{m}
\left[
\phi_n(s_j)-\phi_M(s_j;\alpha,\sigma)
\right]^2
\]

\begin{figure}[H]
	\centering
	\subfigure[Boxplots for different values of $\alpha$.]{\includegraphics[width=0.48\textwidth]{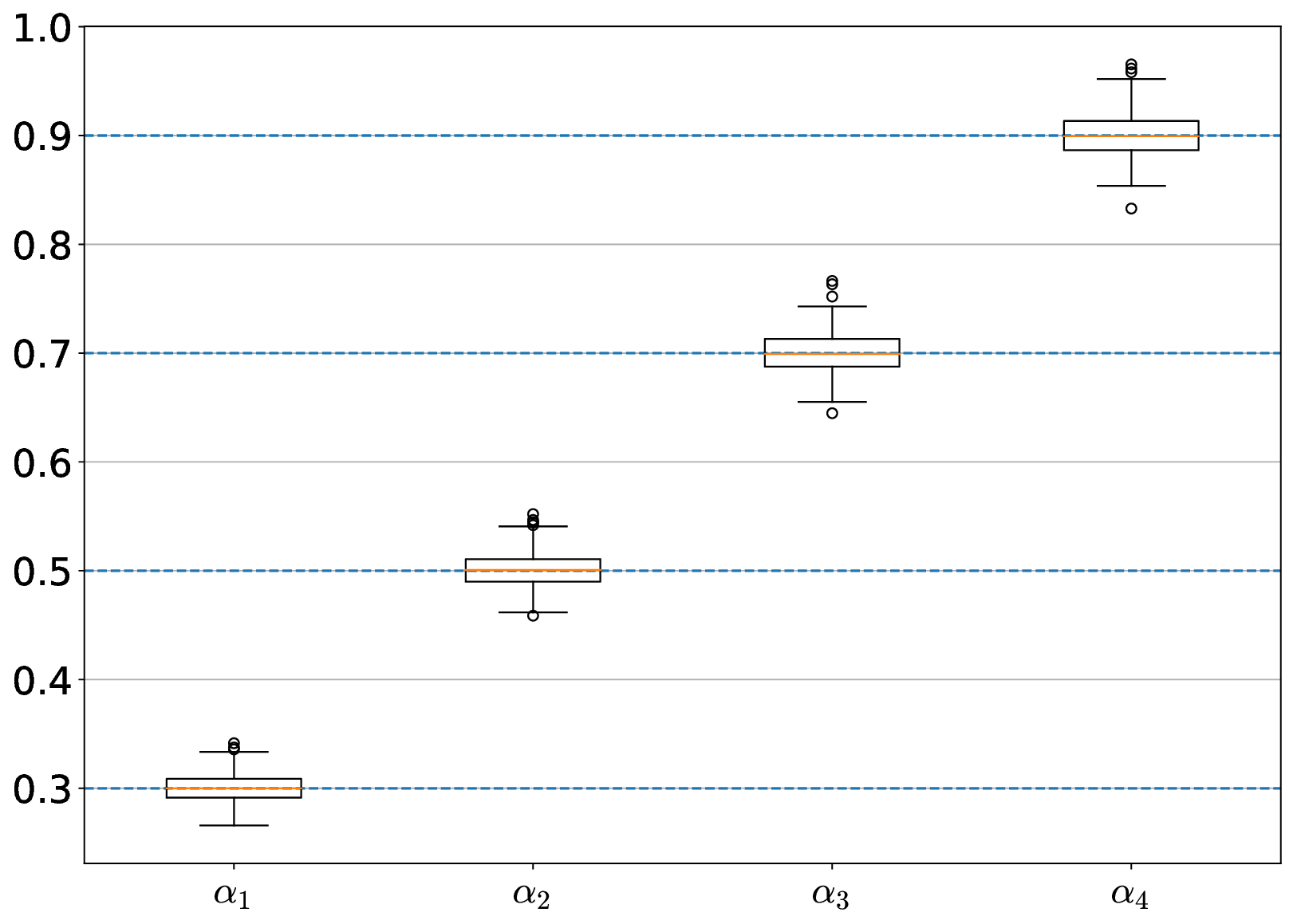}} \hfill
	\subfigure[Boxplots for different values of $\sigma$.]{\includegraphics[width=0.48\textwidth]{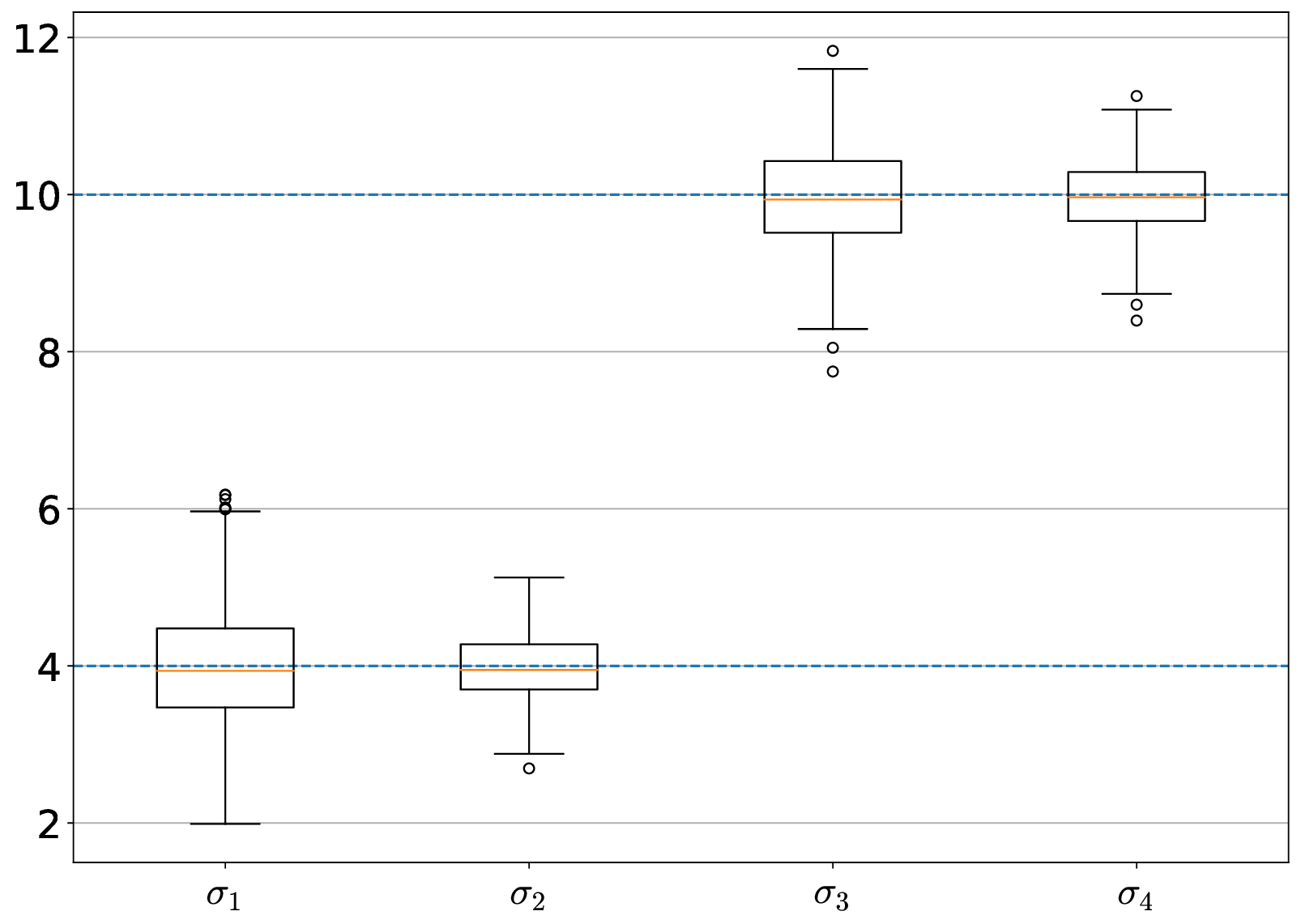}} \hfill
	\caption{The boxplots of ML($\alpha,\sigma$) distribution with true parameters $(\alpha_1, \sigma_1) =(0.3,4)$, ($\alpha_2, \sigma_2)= (0.5, 4)$, ($\alpha_3, \sigma_3)= (0.7, 10)$ and ($\alpha_4, \sigma_4)= (0.9, 10)$. The simulated data has $500$ trajectories from ML($\alpha,\sigma$) each of length $1000.$ }\label{fig6}
\end{figure} 
\begin{table}[ht!]
	\caption{Parameter estimates for four different simulated data.}\label{tab5}
	
	\begin{tabular}{lllll}
		\hline\noalign{\smallskip}
		& $(\alpha_1, \sigma_1) $ & $(\alpha_2, \sigma_2)$ & $(\alpha_3, \sigma_3)$ & $(\alpha_4, \sigma_4)$ \\
		\noalign{\smallskip}\hline\noalign{\smallskip}
		True values & $(0.3, 4)$ & $(0.5, 4)$ & $(0.7, 10)$ & $(0.9, 10)$ \\ 
		Est. values &  $(0.3004,4.0261)$ & $(0.5005, 3.9935)$ & $(0.7001, 9.9744)$ & $(0.8997, 9.9764)$\\
		\noalign{\smallskip}\hline
	\end{tabular}
\end{table}
\begin{table}[ht!]
	\caption{RMSE and MAE of estimated parameters for four simulated data.}\label{tab6}
	
	\begin{tabular}{lllll}
		\hline\noalign{\smallskip}
		& $(\alpha_1, \sigma_1) $ & $(\alpha_2, \sigma_2)$ & $(\alpha_3, \sigma_3)$ & $(\alpha_4, \sigma_4)$ \\
		\noalign{\smallskip}\hline\noalign{\smallskip}
		RMSE & $(0.0126, 0.7625)$ & $(0.0154, 0.4134)$ & $(0.0183, 0.6681)$ & $(0.0198, 0.4583)$ \\ 
		MAE &  $(0.0101, 0.6071)$ & $(0.0124, 0.3330)$ & $(0.0147, 0.5379)$ & $(0.0160, 0.3682)$\\
		
		\noalign{\smallskip}\hline
	\end{tabular}
\end{table}
We simulate four different data sets using the Laplace transform based method given in \cite{Ridout2009} and apply the empirical Laplace method for estimation. The true values and estimated values are tabulated in Table \ref{tab5}. Further, root mean squared error (RMSE) and mean absolute error (MAE) for estimated parameters are given in Table \ref{tab6}. The numerical results imply that the accuracy of estimation method is good. The boxplots for all the four sets are shown in Fig. \ref{fig6} which show that all the estimates are accumulated near the true values of parameters.

\subsection{Real data application} 
To exemplify the application of ML($\alpha,\sigma$) distribution, we consider the data set consisting of oil futures high frequency trading (HFT) data. The HFT data for oil futures is collected from July 9, 2007, to August 9, 2007, comprising a total of $951,091$ observations captured during market opening hours. Inter-arrival times, measured in seconds, have been recorded for both up-ticks and down-ticks. In Poisson process, the inter-arrival time is assumed to be exponentially distributed, whereas for fractional Poisson process the inter-arrival time is from ML($\alpha, \sigma$) distribution. We apply the empirical Laplace transform method and obtain the estimated parameters for positive jumps as $\hat{\alpha} = 0.6463, \hat{\sigma} = 0.5640$ and for negative jumps as $\hat{\alpha} = 0.6269, \hat{\sigma} = 0.5137.$ 
\begin{figure}[H]
	\centering
	\subfigure[Log survival function for positive jumps.]{\includegraphics[width=0.49\textwidth]{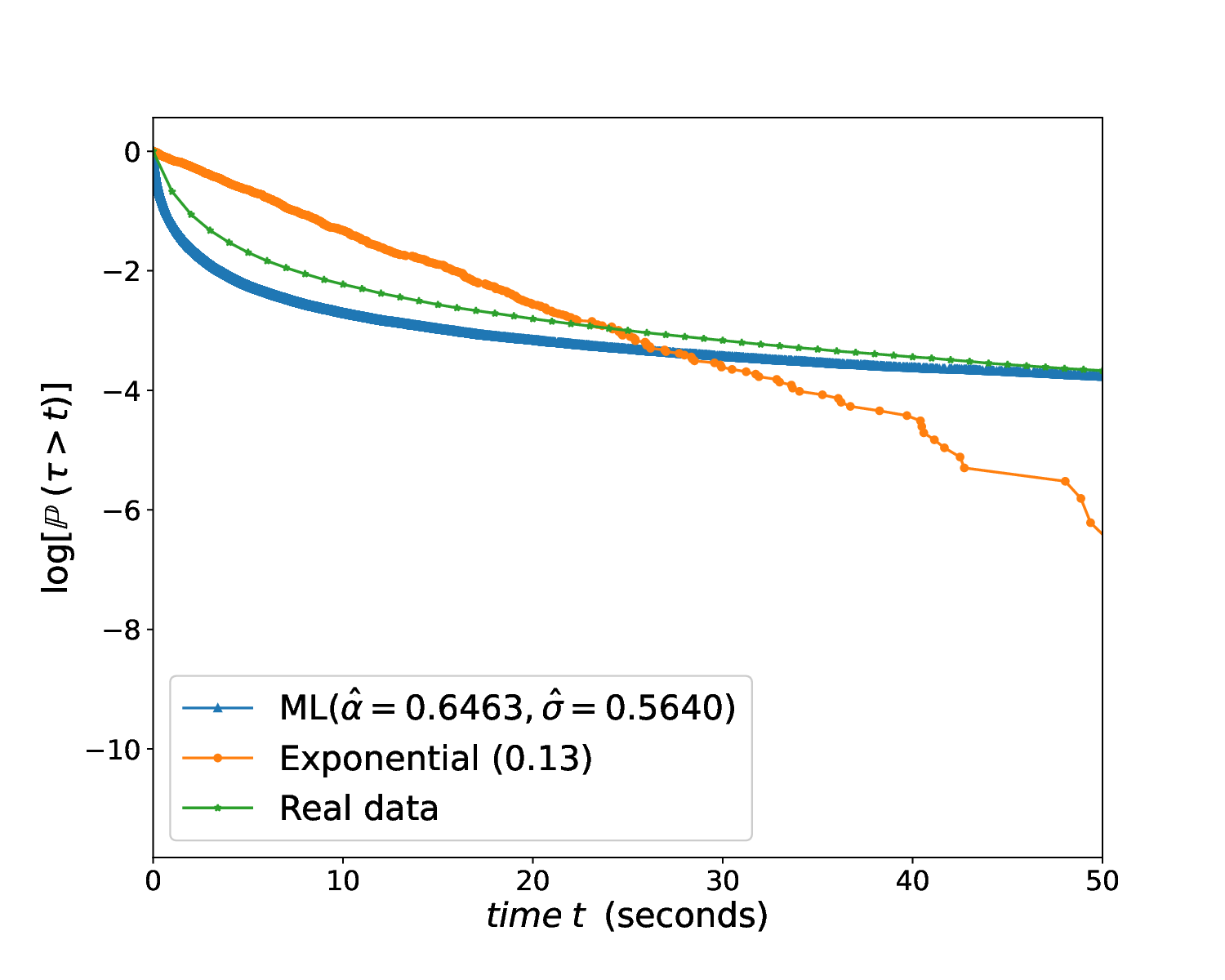}}\hfill
	\subfigure[Log survival function for negative jumps.]{\includegraphics[width=0.49\textwidth]{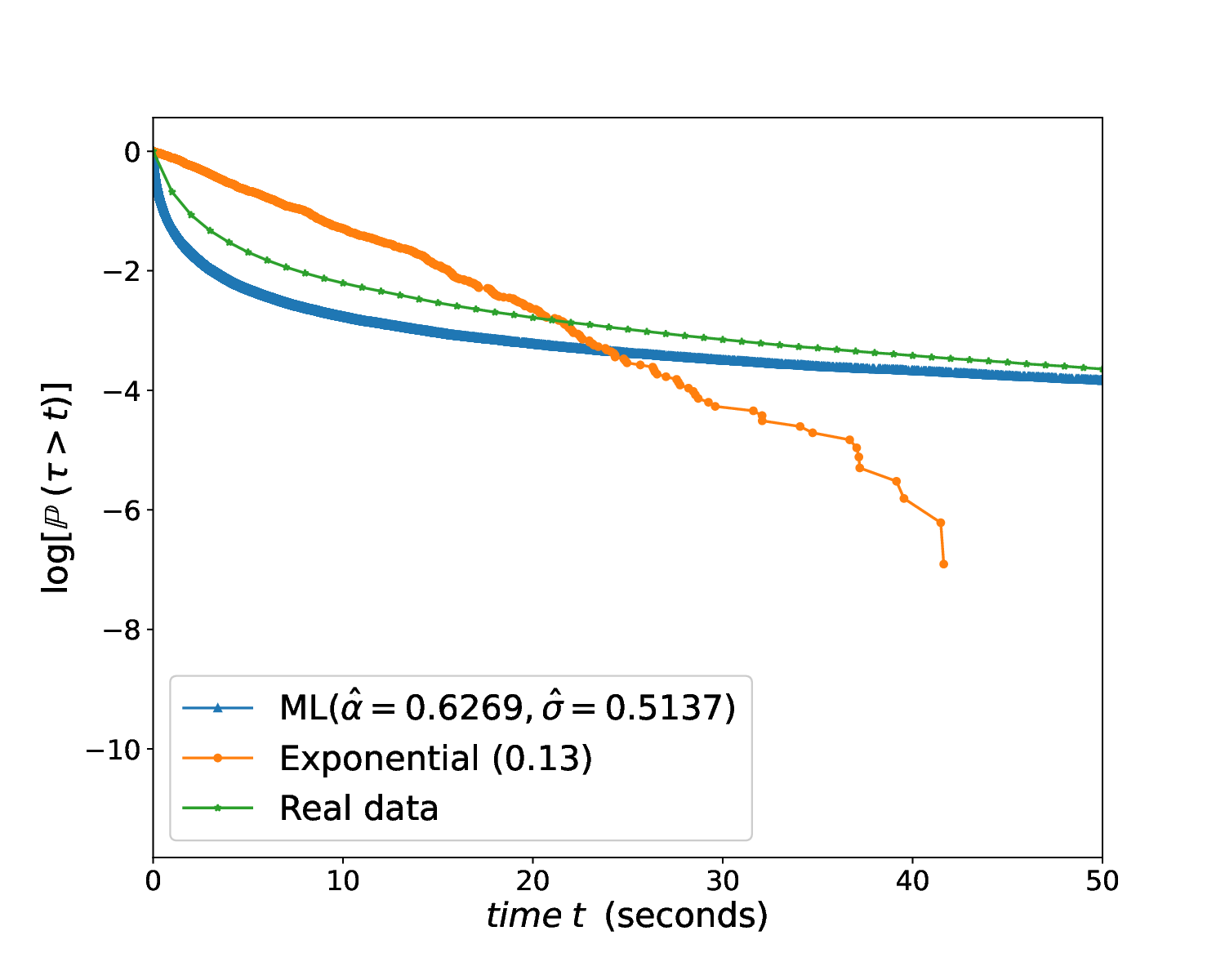}} 
	\caption{Survival functions for oil futures high frequency trading data compared with the corresponding exponential distribution's survival functions.}\label{fig:ML_HFT}
\end{figure}

The Fig. \ref{fig:ML_HFT} illustrates the logarithmic-scale survival function of inter-arrival times in the data, compared with the exponential distribution. We observe that the data is not aligned with exponential inter-arrival time, as is the case for the Poisson process. It is thus evident that, arrival times of the data is very well modeled by ML($\alpha, \sigma$) distribution. Recently, the same behavior was empirically observed by \cite{Arun2020}, while studying the fractional Poisson process.

Next section deals with the estimation of three-parameter Mittag-Leffler distribution using empirical Laplace transform.

\section{Mittag-Leffler($\alpha, \sigma,\gamma$) or Prabhakar distribution}\label{sec:ML4}
The Mittag-Leffler function is widely regarded as the leading function in fractional calculus. Being the generalized version of Mittag-Leffler function, the distribution using Prabhakar function will provide greater flexibility in modeling complex systems. The study of Prabhakar distribution or Prabhakar processes will open new avenues for developing probabilistic models that better capture real-world systems with long-range dependence, heavy tails, or power-law behavior.
The three parameter Mittag-Leffler function, also known as Prabhakar function has the following infinite series representation \cite{Prabh1}:
\begin{align}
	E^{\gamma}_{\alpha, \sigma}(z) = \dfrac{1}{\Gamma(\gamma)}\sum_{k=0}^{\infty} \dfrac{\Gamma(k+\gamma)z^k}{k!\Gamma(\alpha k+\sigma)},\,\,\, z\in \mathbb{C}, \,\alpha, \sigma, \gamma \in \mathbb{C}.\label{eq_prabh1}
\end{align}
When $\gamma=1$, $E$ reduces to two parameter Mittag-Leffler function, while, for $\sigma=1$ and $\gamma=1$, we obtain the classical Mittag-Leffler function. The Prabhakar Levy process is defined by considering a Lévy measure expressed by the means of the function in Eq. \eqref{eq_prabh1}, for $\alpha \in (0,1], \sigma \in [1, 1+\alpha \gamma]$ and $\gamma>0$. The parameter $\alpha \in (0,1)$ is the index of stability, $\sigma>0$ is the scale parameter and $\gamma>0$ is the shape parameter or mixing parameter. We use the results provided by \cite{Prabh1} on Prabhakar Levy process to study the Prabhakar distribution. Note that, the Prabhakar distribution can be defined as the marginal distribution of the Prabhakar Lévy subordinator evaluated at $t=1$.
One of the important results related to the Laplace exponent of Prabhakar process is as follows:
\vspace{0.4cm}
\begin{result}
	The Laplace exponent of the Prabhakar process reads,
	\begin{align}
		\Psi_{M_{\alpha,\sigma}^{\gamma}}(u) = \dfrac{u^{\alpha\gamma-\sigma+1}\Gamma(\gamma)}{\lambda^{\gamma}\Gamma(\gamma+ \frac{1-\sigma}{\alpha}+1)} {}_2F_1\left(\gamma, \gamma+\frac{1-\sigma}{\alpha}, 1+\gamma \frac{1-\sigma}{\alpha}; -\frac{u^{\alpha}}{\lambda}\right), \label{eq_prabh2}
	\end{align}
	where ${}_2F_1(a,b,c:z)$ is the Gauss hypergeometric function \cite{Prabh1}.
\end{result}
\vspace{0.4cm}

For a subordinator $T(t)$ with Laplace exponent $\Psi(s)$, the Laplace transform is
	
	\[
	E(e^{-sT(t)}) = \exp(-t\Psi(s)).
	\]
	
	Hence the marginal distribution of the process at $t=1$ has Laplace transform
	
	\[
	\phi_M(s) = \exp(-\Psi(s)).
	\]
For Prabhakar distribution, we find the Laplace transform using the Laplace exponent from Eq. \eqref{eq_prabh2} by substituting $\lambda=1$. We obtain the following result:
\vspace{0.4cm}

\begin{corollary}
	The Laplace transform of Prabhakar distribution $M_{\alpha,\sigma}^{\gamma}$ is,
	\begin{align}\label{eq_prabh3}
		\phi_{M}(s) = \exp\left(- \frac{s^{\alpha\gamma-\sigma+1}\Gamma(\gamma)}{\Gamma(\gamma+1+\frac{1-\sigma}{\alpha})} {}_2F_1\left(\gamma, \gamma+\frac{1-\sigma}{\alpha}, 1+\gamma \frac{1-\sigma}{\alpha}, -s^{\alpha}\right)\right),
	\end{align}
	where ${}_2F_1(a,b,c;z)$ is the Gauss hypergeometric function.   
\end{corollary}
\vspace{0.4cm}
The data is simulated using the Laplace transform defined in Eq. \eqref{eq_prabh3} by the method proposed by \cite{Ridout2009}. We implement the empirical Laplace transform method to estimate the parameters. The simulation study on three different data sets of $100$ trajectories each of length $1000$ is performed. The true parameters values and estimated values are given in Table \ref{tab7}. The RMSE and MAE are provided in Table \ref{tab8}. From the boxplots in Fig. \ref{fig7} and the estimated data in Table \ref{tab7}, it is evident that the method provides the good estimates. Observe that there is slight deviation of the $\gamma$ estimates from the true $\gamma$ values, but the overall results are satisfactory.
\begin{figure}[H]
	\centering
	\subfigure[Boxplots for different values of $\alpha$.]{\includegraphics[width=0.32\textwidth]{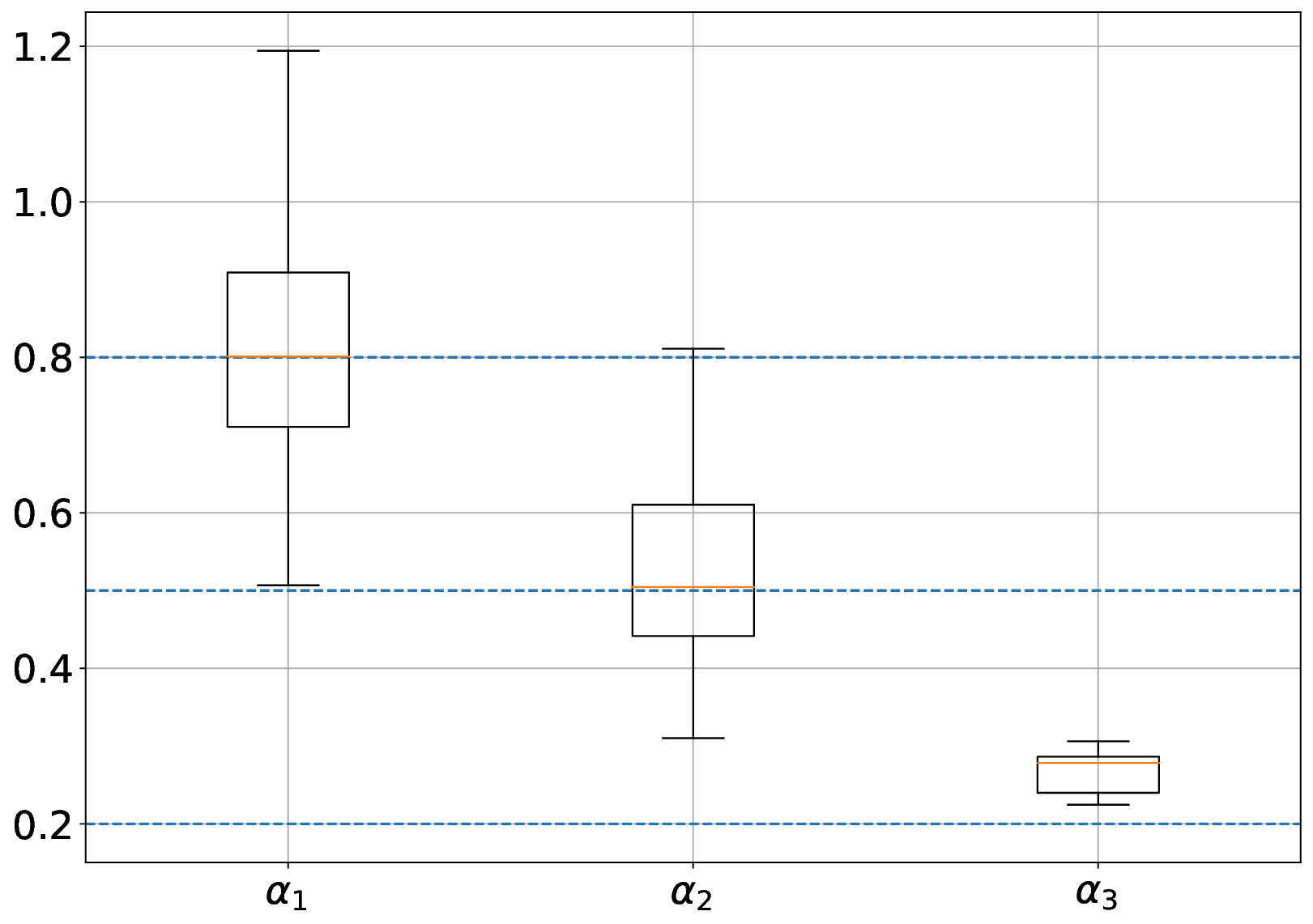}} \hfill
	\subfigure[Boxplots for different values of $\sigma$.]{\includegraphics[width=0.32\textwidth]{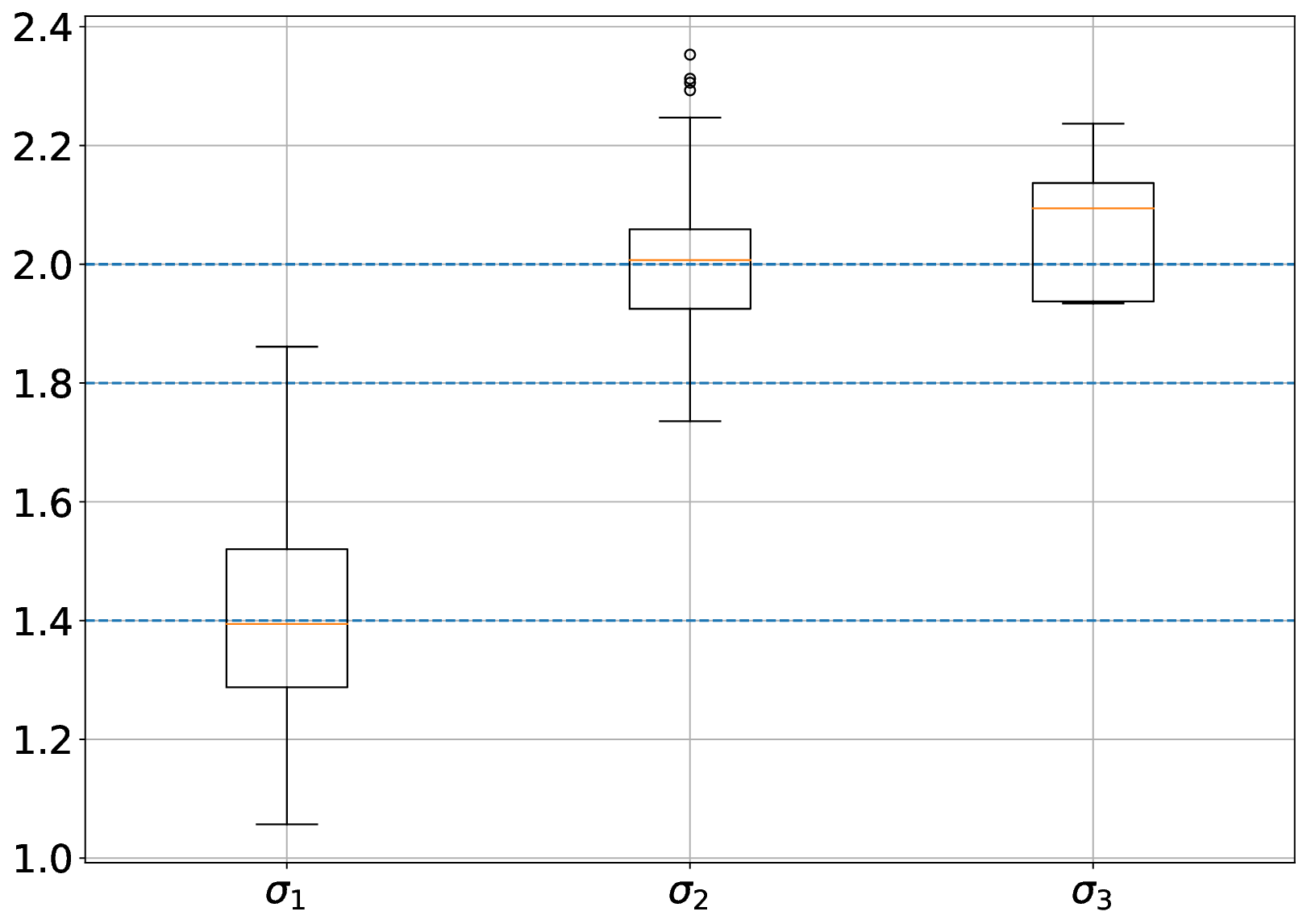}} \hfill
	\subfigure[Boxplots for different values of $\gamma$.]{\includegraphics[width=0.32\textwidth]{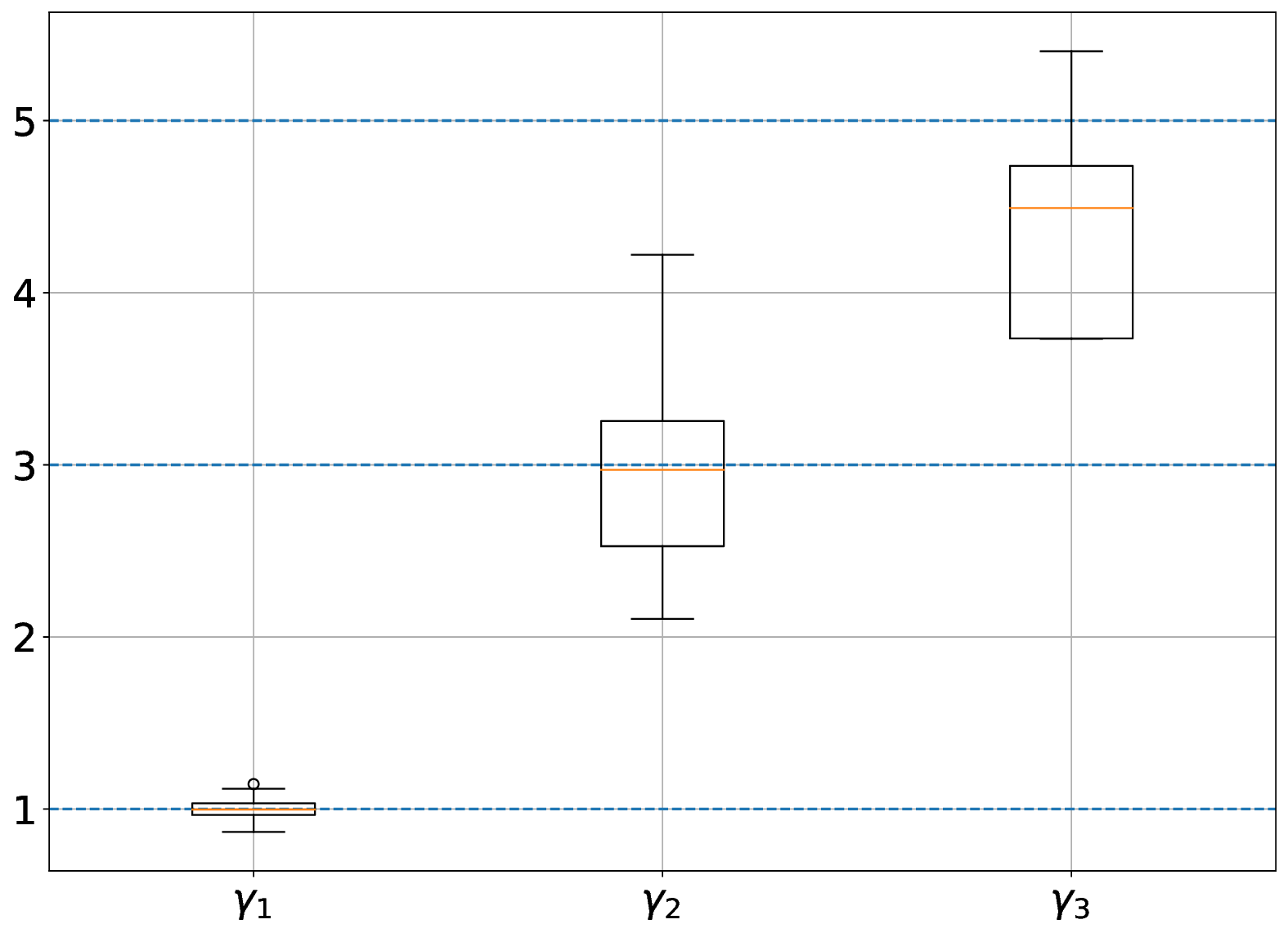}} \hfill
	\caption{The boxplots of ML($\alpha,\sigma, \gamma$) distribution with true parameters $(\alpha_1, \sigma_1, \gamma_1) =(0.8,1.4,1)$, ($\alpha_2, \sigma_2, \gamma_2)= (0.5, 2,3)$, and ($\alpha_3, \sigma_3, \gamma_3)= (0.2,1.8,5)$. The simulated data has $100$ trajectories from ML($\alpha,\sigma, \gamma$) each of length $1000.$ }\label{fig7}
\end{figure} 

\begin{table}[ht!]
	\centering
	\caption{Parameter estimates for three different simulated data.}\label{tab7}
	
	\begin{tabular}{llll}
		\hline{\smallskip}
		& $(\alpha_1, \sigma_1, \gamma_1) $ & $(\alpha_2, \sigma_2, \gamma_2)$ & $(\alpha_3, \sigma_3,\gamma_3)$ \\
		\noalign{\smallskip}\hline\noalign{\smallskip}
		True values & $(0.8, 1.4, 1)$ & $(0.5, 2, 3)$ & $(0.2, 1.8, 5)$\\ 
		Est. values &  $(0.809,1.411, 0.998)$ & $(0.523, 2.002, 2.955)$ & $(0.269, 2.068, 4.392)$ \\
		\noalign{\smallskip}\hline
	\end{tabular}
\end{table}
\begin{table}[ht!]
	\centering
	\caption{RMSE and MAE of estimated parameters for three simulated data.}\label{tab8}
	
	\begin{tabular}{llll}
		\hline{\smallskip}
		& $(\alpha_1, \sigma_1, \gamma_1) $ & $(\alpha_2, \sigma_2, \gamma_2)$ & $(\alpha_3, \sigma_3, \gamma_3)$ \\
		\noalign{\smallskip}\hline\noalign{\smallskip}
		RMSE & $(0.135, 0.160, 0.049)$ & $(0.116, 0.129,0.542)$ & $(0.073, 0.283, 0.764)$  \\ 
		MAE &  $(0.110, 0.129, 0.039)$ & $(0.094, 0.099, 0.444)$ & $(0.069, 0.269,0.626)$ \\
		
		\noalign{\smallskip}\hline
	\end{tabular}
\end{table}

\subsection{Subordination Representation and Simulation}
In this subsection we present some alternate ways to study the Prabhakar distribution on the restricted parameter domain. An important probabilistic 
interpretation of the Prabhakar Lévy process can be obtained using the concept of subordination, which we discuss in the subsequent subsection.
\begin{proposition}
Let $0<\alpha \leq 1$, $\sigma=\alpha\gamma$, $\gamma>0$, and $\lambda>0$, and define
\[
f(x)=\lambda^{\gamma}x^{\sigma-1}
E_{\alpha,\sigma}^{\gamma}(-\lambda x^{\alpha}),
\qquad x>0,
\]
where the three-parameter Mittag-Leffler (Prabhakar) function is
\[
E_{\alpha,\sigma}^{\gamma}(z)
= \sum_{k=0}^{\infty}
\frac{(\gamma)_k z^k}{k!\Gamma(\alpha k+\sigma)} .
\]

Then, for $s>0$, the Laplace transform of $f(x)$ exists and is given by
\begin{align}\label{eq_prabh4}
\mathcal{L}\{f(x)\}(s)
= \int_{0}^{\infty} e^{-sx} f(x)\,dx
= \left(\frac{\lambda}{\lambda+s^\alpha}\right)^\gamma .
\end{align}
\end{proposition}

\begin{proof}
From the Laplace transform formula for the Prabhakar function
(see \cite{Guisti2020}[Eq. (4.17)]), for $\Re(s)>0$ and 
$|s|>|\lambda|^{1/\alpha}$,
\[
\int_{0}^{\infty} e^{-sx}
x^{\sigma-1}
E_{\alpha,\sigma}^{\gamma}(-\lambda x^\alpha)\,dx
=
\frac{s^{\alpha\gamma-\sigma}}
{(s^\alpha+\lambda)^\gamma}.
\]

Multiplying both sides by $\lambda^\gamma$ yields
\[
\mathcal{L}\{f(x)\}(s)
=
\frac{\lambda^\gamma s^{\alpha\gamma-\sigma}}
{(s^\alpha+\lambda)^\gamma}.
\]
Using condition $\sigma=\alpha\gamma$, we get 
\begin{align}
\mathcal{L}\{f(x)\}(s) = \left(\frac{\lambda}{\lambda+s^\alpha}\right)^\gamma. \label{Leq}
\end{align}
The conditions $0 \leq \alpha \leq 1$, $\sigma=\alpha\gamma$, $\gamma>0$, and $s>0$
ensure convergence of the integral and justify the use of the Laplace transform formula.
\end{proof}


Hence, we say, a non-negative random variable $X$ is said to follow the Prabhakar distribution if its density is given by
\begin{equation*}
f_X(x) = \lambda^{\gamma} x^{\sigma-1}
E_{\alpha,\sigma}^{\gamma}(-\lambda x^{\alpha}),
\qquad x>0,
\end{equation*}
with parameters $0<\alpha\le1$, $\sigma=\alpha\gamma$, $\gamma>0$, and $\lambda>0$.

The Laplace transform of $X$ is
\begin{equation*}
\mathbb{E}[e^{-sX}]
=
\left(\frac{\lambda}{\lambda+s^{\alpha}}\right)^{\gamma}.
\end{equation*}

\subsubsection{Subordination Representation}
A subordinator is a non-decreasing Lévy process that is often used to construct new processes through random time changes.

Let $S_{\alpha}(t)$ be a stable subordinator with stability index $0<\alpha<1$ and Laplace transform \cite{bertoin1999subordinators}

\[
E\left[e^{-sS_{\alpha}(t)}\right] = \exp(-t s^{\alpha}), \qquad s>0.
\]

Further, let $G_{\gamma}(t)$ denote a gamma subordinator with
shape parameter $\gamma>0$ whose Laplace transform \cite{bertoin1999subordinators} is

\[
E\left[e^{-sG_{\gamma}(t)}\right]
=
(1+s)^{-\gamma t}.
\]

The Prabhakar Lévy process can be represented as a scaled
subordinated process obtained by evaluating the stable
subordinator at the random time $G_{\gamma}(t)$ \cite{Prabh1}. In particular,

\[
M(t)=\lambda^{-1/\alpha}S_{\alpha}(G_{\gamma}(t)).
\]

Using the law of total expectation, the Laplace transform of the process is obtained as

\[
E\left[e^{-sM(t)}\right] =
E\left[\exp\left(-G_{\gamma}(t)
\left(\frac{s}{\lambda^{1/\alpha}}\right)^{\alpha}\right)\right].
\]

Since $G_{\gamma}(t)$ is gamma distributed, we obtain

\[
E\left[e^{-sM(t)}\right] =
\left(1+\frac{s^{\alpha}}{\lambda}\right)^{-\gamma t}.
\]

Hence the marginal distribution of the Prabhakar process at $t=1$ has Laplace transform

\[
\phi_M(s)=\left(1+\frac{s^{\alpha}}{\lambda}\right)^{-\gamma}
= \left(\frac{\lambda}{\lambda + s^{\alpha}}\right)^{\gamma}.
\]

This approach also leads to the same form of Laplace transform as in Eq. \eqref{Leq}. This subordination representation provides a convenient probabilistic interpretation of the Prabhakar distribution
and forms the basis for efficient simulation procedures. For more details on the construction and properties of
Prabhakar Lévy processes, see \cite{Prabh1} and survey by Giusti et al. \cite{Guisti2020}. 
\subsubsection{Simulation Algorithm}
First we need to generate of positive $\alpha$-stable random variable. A positive $\alpha$-stable random variable $S_{\alpha}$ can be generated using the method proposed by Kanter \cite{kanter1975}. Let

\[
U \sim \text{Uniform}(0,\pi), \qquad
W \sim \text{Exponential}(1).
\]

Then the positive $\alpha$-stable random variable is given by
\begin{equation}
S_{\alpha} = \frac{\sin(\alpha U)}{(\sin U)^{1/\alpha}}
\left(
\frac{\sin((1-\alpha)U)}{W}
\right)^{\frac{1-\alpha}{\alpha}}.
\end{equation}

\begin{algorithm}[H]
\caption{Simulation of the Prabhakar Distribution}
\begin{algorithmic}[1]
\Require Parameters $\alpha \in (0,1]$, $\gamma>0$, $\lambda>0$
\Ensure Random sample $X$ from the Prabhakar distribution

\State Generate $G \sim \Gamma(\gamma,\lambda)$
\State Generate $U \sim \text{Uniform}(0,\pi)$
\State Generate $W \sim \text{Exponential}(1)$
\State Compute the positive $\alpha$-stable random variable
\[
S_{\alpha} = \frac{\sin(\alpha U)}{(\sin U)^{1/\alpha}}
\left(\frac{\sin((1-\alpha)U)}{W}
\right)^{\frac{1-\alpha}{\alpha}}
\]
\State Compute
\[
X = G^{1/\alpha} S_{\alpha}
\]
\State Return $X$
\end{algorithmic}
\end{algorithm}

In the next section, we study the autoregressive model of order $1$ (AR$(1)$) with ML distribution and also obtain the integral form of probability density function of the innovation terms.
\section{Autoregressive models}\label{sec:ML3}
Consider the stationary first order autoregressive model (AR($1$)) as,
\begin{align}
	Y_t = \rho Y_{t-1} + \epsilon_t,\label{ARML}
\end{align} 
where $0\leq \rho<1, \, t = 1,2,\hdots, n.$ In the following result we study the distribution of innovation terms and $\{Y_t\}$ using the Laplace transform of ML$(\alpha, 1)$ as $\phi_Y(s) = \dfrac{1}{1+s^\alpha}.$
\begin{proposition}
	Consider the stationary AR$(1)$ model defined in Eq. \eqref{ARML}, with the marginal distribution of $Y_t$ from ML$(\alpha,1)$ then the innovation terms $\{\epsilon_t\}$ have distribution with Laplace transform $\phi_\epsilon(s) = \dfrac{1+(\rho s)^{\alpha}}{1+s^{\alpha}}$. Conversely, if the innovation terms are from ML$(\alpha,1)$ distribution, then the marginal distribution of $Y_t$ has Laplace transform $\phi_Y(s) = \displaystyle\dfrac{1}{\prod_{i=0}^\infty(1+(\rho^i s)^\alpha)}.$
\end{proposition}
\begin{proof}
	The Laplace transform of ML$(\alpha,1)$ marginals is $\phi_Y(s) = \dfrac{1}{1+s^\alpha}.$
	Then the Laplace transform of the model will be,
	\begin{align*}
		\phi_Y(s) &= \phi_Y(\rho s) \phi_\epsilon(s)\implies
		\dfrac{1}{1+s^\alpha} = \dfrac{1}{1+(\rho s)^\alpha} \phi_\epsilon(s),\\
		\phi_\epsilon(s) &= \dfrac{1+(\rho s)^\alpha}{1+ s^\alpha}.
	\end{align*}
	
	Now, assume that the innovation terms are from ML$(\alpha,1)$ distribution. We rewrite the stationary AR$(1)$ model in the form of MA$(\infty)$ using backward shift operator $B(Y_t) = Y_{t-1}$ as follows:
	\begin{align*}
		Y_t - \rho Y_{t-1}  &= \epsilon_t \implies (1-\rho B) Y_t = \epsilon_t,\\
		Y_t &= \dfrac{\epsilon_t}{1-\rho B} = \displaystyle\sum_{i=0}^{\infty}\rho^i\epsilon_{t-i}; 0 \leq\rho<1.
	\end{align*}
	We take the Laplace transform on both sides of above equation and obtain the desired result,
	$$\phi_Y(s) = \prod_{i=0}^{\infty}\phi_{\epsilon_{t-i}}(\rho^i s) = \dfrac{1}{\prod_{i=0}^{\infty}(1+(\rho^i s)^\alpha)}.$$
	
\end{proof}
\begin{remark}
The Laplace transform $\phi_\epsilon(s) = \dfrac{1+(\rho s)^\alpha}{1+ s^\alpha}$ can be rewritten as 
\begin{align*}
\phi_\epsilon(s) &= \rho^\alpha + \dfrac{1-\rho^\alpha}{1+ s^\alpha}.
\end{align*}
This shows that the distribution of innovation terms is a mixed distribution rather than a continuous distribution. 
\end{remark}
\begin{theorem}
	For stationary AR($1$) model with ML($\alpha,1$) marginals, the probability density function of the continuous part of innovation terms is
	\begin{equation}
		f_\epsilon(x)= \frac{1}{\pi} \int_{0}^{\infty} \frac{ e^{-xy}y^{\alpha}\sin(\pi\alpha)\{1-  \rho^{\alpha}\}}{1+y^{2\alpha} +  2y^{\alpha}\cos(\pi\alpha)} \,dy.\label{AR_pdf}
	\end{equation}
\end{theorem}
\begin{proof}
	The Laplace transform of error terms is $\phi_\epsilon(s) = \dfrac{1+(\rho s)^{\alpha}}{1+s^{\alpha}}$. Consider the function $g(s)= \frac{1+\rho^\alpha s^\alpha}{1+s^\alpha}.$ 
	The branch point for $g(s)$ is $s_1=(0,0)$. We use complex inversion formula to compute the Laplace inverse of $g(s)$ which in turn gives the probability density function $f_\epsilon(x)$ of innovation terms
	$$f_{\epsilon}(x) = \frac{1}{2\pi\iota}\int_{x_0-i\infty}^{x_0+i\infty}e^{sx}g(s)\,ds,$$ where $x_0>a$ is chosen such that the integrand is analytic for $Re(s)>a.$ 
	Consider the contour $ABCDEF$ in Fig. \ref{fig5} with branch point $P=(0,0)$, circular arcs $AB$ and $EF$ of radius $R$, arc $CD$ of radius $r$, line segments $BC$ and $DE$, are parallel to $x-$axis and $AF$ is line segment from $x_0-iy$ to $x_0+iy$. For closed curve $C$, Cauchy residue theorem states that \cite{Schiff1999}, $$\frac{1}{2\pi\iota}\oint_C e^{sx}g(s)\,ds = 0.$$
	\begin{figure}[ht!]
		\centering
		\includegraphics[scale=0.65]{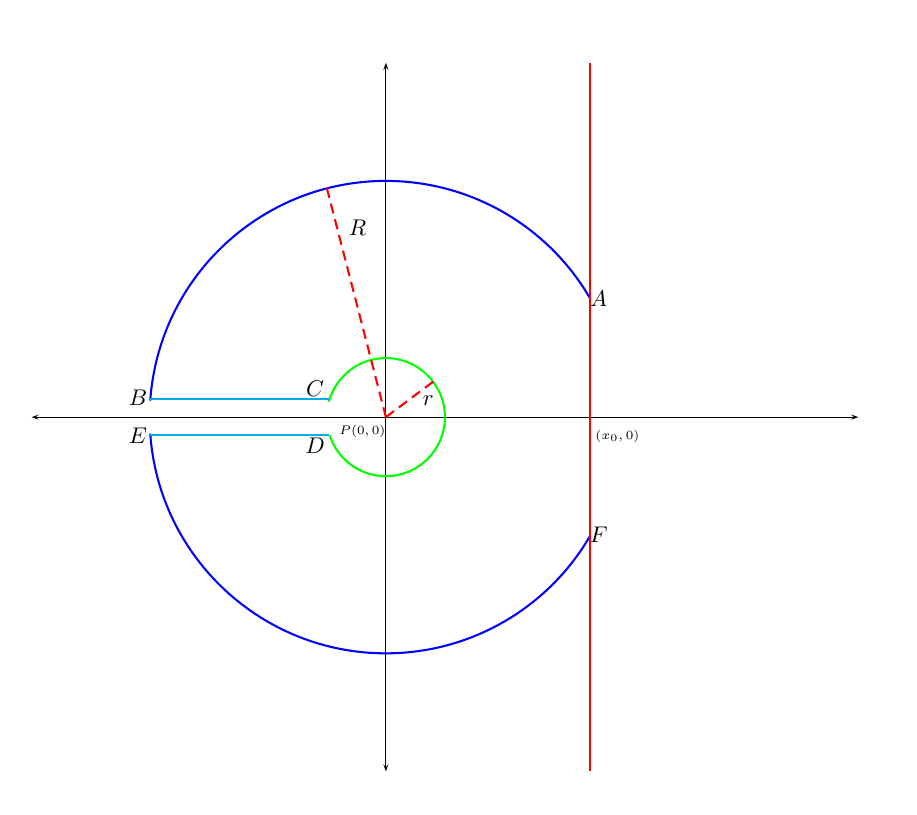}
		\caption{Contour plot with branch point at $O=(0,0)$ in anti-clockwise direction. }
		\label{fig5}
	\end{figure}
	In limiting case, for contour $ABCDEF$, the integral on circular arcs $AB$ and $EF$ tend to $0$ as $R\to \infty.$ The integral over $CD$ also tends to $0$ as $r\to 0.$ We need to compute the following integral: 
	\begin{align}
		\frac{1}{2\pi\iota}\int_{x_0-\iota \infty}^{x_0+\iota\infty} e^{sx}g(s)\,ds &= - \frac{1}{2\pi\iota}\int_{BC}e^{sx}g(s)\,ds - \frac{1}{2\pi\iota}\int_{DE}e^{sx}g(s)\,ds.\label{int_eq1}
	\end{align}

	\noindent For $BC$, let $s=ye^{\iota\pi}$, then $ds=-dy$,
	\begin{align*}
		\int_{BC}\frac{e^{sx}(1+ \rho^\alpha s^\alpha)}{1+s^\alpha}\,ds &= \int_{-R}^{-r}\frac{e^{sx}(1+ \rho^\alpha s^\alpha)}{1+s^\alpha}\,ds\\
		&= -\int_{R}^{r}\frac{e^{-yx}(1+ \rho^\alpha y^\alpha e^{\iota\pi\alpha})}{1+y^\alpha e^{\iota\pi\alpha}}\,dy\\
		&= \int_{r}^{R}\frac{e^{-yx}(1+ \rho^\alpha y^\alpha e^{\iota\pi\alpha})}{1+y^\alpha e^{\iota\pi\alpha}}\,dy.
	\end{align*}
	For $DE$, let $s=ye^{-\iota\pi}$, then $ds=-dy$.
	\begin{align*}
		\int_{DE}\frac{e^{sx}(1+ \rho^\alpha s^\alpha)}{1+s^\alpha}\,ds &= \int_{-r}^{-R}\frac{e^{sx}(1+ \rho^\alpha s^\alpha)}{1+s^\alpha}\,ds\\
		&= -\int_{r}^{R}\frac{e^{-yx}(1+ \rho^\alpha y^\alpha e^{-\iota\pi\alpha})}{1+y^\alpha e^{-\iota\pi\alpha}}\,dy.
	\end{align*}
	Now, combining both integrals, we get
	\begin{align*}
		\frac{1}{2\pi\iota}\Biggl[-\displaystyle\int_{BC}\frac{e^{sx}(1+ \rho^\alpha s^\alpha)}{1+s^\alpha}\,ds &+ \displaystyle\int_{DE}\frac{e^{sx}(1+ \rho^\alpha s^\alpha)}{1+s^\alpha}\,ds\Biggr] \\
		&=\frac{1}{2\pi\iota}\int_{r}^{R}{e^{-xy}}\Biggl\{-\dfrac{1+\rho^\alpha y^\alpha e^{\iota\pi\alpha}}{1+ y^\alpha e^{\iota\pi\alpha}} + \dfrac{1+\rho^\alpha y^\alpha e^{-\iota\pi\alpha}}{1+ y^\alpha e^{-\iota\pi\alpha}}\Biggl\}\,dy\\ 
		&= \frac{1}{2\pi\iota}\int_{r}^{R}\dfrac{e^{-xy}\{y^\alpha(-e^{-\iota\pi\alpha} + e^{\iota\pi\alpha}) - \rho^{\alpha}y^\alpha(e^{\iota\pi\alpha}-e^{-\iota\pi\alpha})\}}{1+ y^{2\alpha} + y^\alpha(e^{\iota\pi\alpha} + e^{-\iota\pi\alpha})}\,dy\\
		&= \frac{1}{2\pi\iota}\int_{r}^{R}\dfrac{e^{-xy}\{-2\iota\sin(\pi\alpha)y^\alpha \rho^\alpha + 2\iota \sin(\pi\alpha)y^\alpha\}}{1+y^{2\alpha} + 2\cos(\pi\alpha)y^\alpha}\,dy\\
		&= \frac{1}{\pi}\int_{r}^{R}\dfrac{e^{-xy}(1-\rho^\alpha)(y^\alpha\sin(\pi\alpha))}{1+y^{2\alpha} + 2\cos(\pi\alpha)y^\alpha} \,dy.\label{int_eq_5}
	\end{align*}  
	Take limits $r\to 0, R\to\infty$ and we get the desired result.
\end{proof}
\vspace{0.5cm}
Time reversibility property of the model helps in finding the hidden patterns in the data. We obtain that the model is not time reversibile in the result below. 
\vspace{0.5cm}
\begin{proposition}
	The AR($1$) model defined in Eq. \eqref{ARML} with $0< \rho<1$ and marginals $\{Y_t\}$ from ML$(\alpha,1)$ is not time reversible.
\end{proposition}

\begin{proof}
	For the model defined in Eq. \eqref{ARML}, the joint Laplace transform is,
	\begin{align*}
		\phi_{Y_{t-1}, Y_{t}}(s_1,s_2) &= \mathbb{E}(\exp{-(s_1Y_{t-1}+s_2Y_{t})})= \mathbb{E}(\exp{-(s_1Y_{t-1}+s_2(\rho Y_{t-1}+\epsilon_t)}))\\
		&= \mathbb{E}(\exp{-(s_1+\rho s_2)Y_{t-1})} \mathbb{E}(\exp{-s_2\epsilon_t})\\
		&= \dfrac{1+(\rho s_2)^\alpha}{(1+s_2^\alpha)(1+(s_1+\rho s_2)^\alpha)}.
	\end{align*}
	The joint Laplace transform is not a symmetric function which imply that the AR($1$) model is not time reversible.
\end{proof}
\subsection{Simulation study}
This subsection deals with the assessment and simulation of the proposed time series AR$(1)$ model. We perform simulation for AR$(1)$ model with marginals from $ML(\alpha, 1)$ which results in innovation terms $\{\epsilon_t\}$ from distribution with Laplace transform $\frac{1+(\rho s)^\alpha}{1+ s^\alpha}$. We use the method proposed by \cite{Ridout2009} to simulate the $N = 500$ iid trajectories of $\{\epsilon_t\}$ each of length $n=1000$. 
We simulate two sets of innovation terms namely, data $1$ with true parameter values $\alpha_1 = 0.4, \rho_1=0.4$ and data $2$ with values $\alpha_2 = 0.6, \rho_2 = 0.8$. Then we generate the $Y_t$ series using relation defined in Eq. \eqref{ARML}. The times series plot and innovation terms for data $1$ and data $2$ are shown in Fig. \ref{fig3}. We apply the empirical Laplace transform to estimate the parameters of AR$(1)$ model. The average of the estimated values of $500$ trajectories are given in Table \ref{tab3}. The RMSE and MAE of the estimation method are in Table \ref{tab4}. The method has less RMSE and MAE and provide the good estimates. The boxplots for the estimates of both the cases are shown in Fig. \ref{fig4} which indicates that the median of the estimates are close to true parameter values. Also the variance of the estimates is less. We can conclude that the empirical Laplace transform method for both the cases perform well on AR$(1)$ model with ML $(\alpha, 1)$ marginals.

\begin{figure}[H]
	\centering
	\subfigure[$Y_t$ for $\alpha_1 = 0.4,\rho_1 = 0.4.$]{\includegraphics[width=0.49\textwidth]{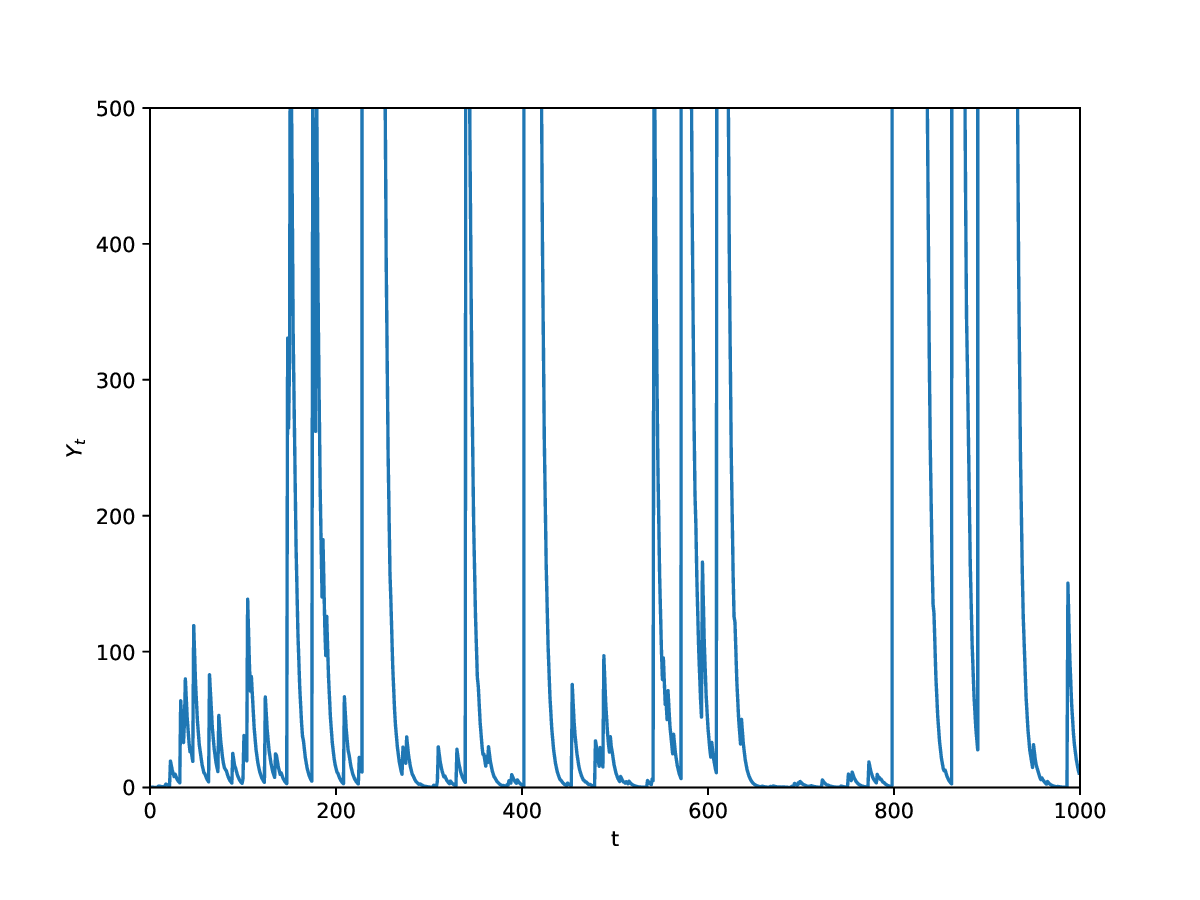}}
	\subfigure[$\epsilon_t$ for ${\alpha_1} = 0.4, \rho_1 = 0.4.$]{\includegraphics[width=0.49\textwidth]{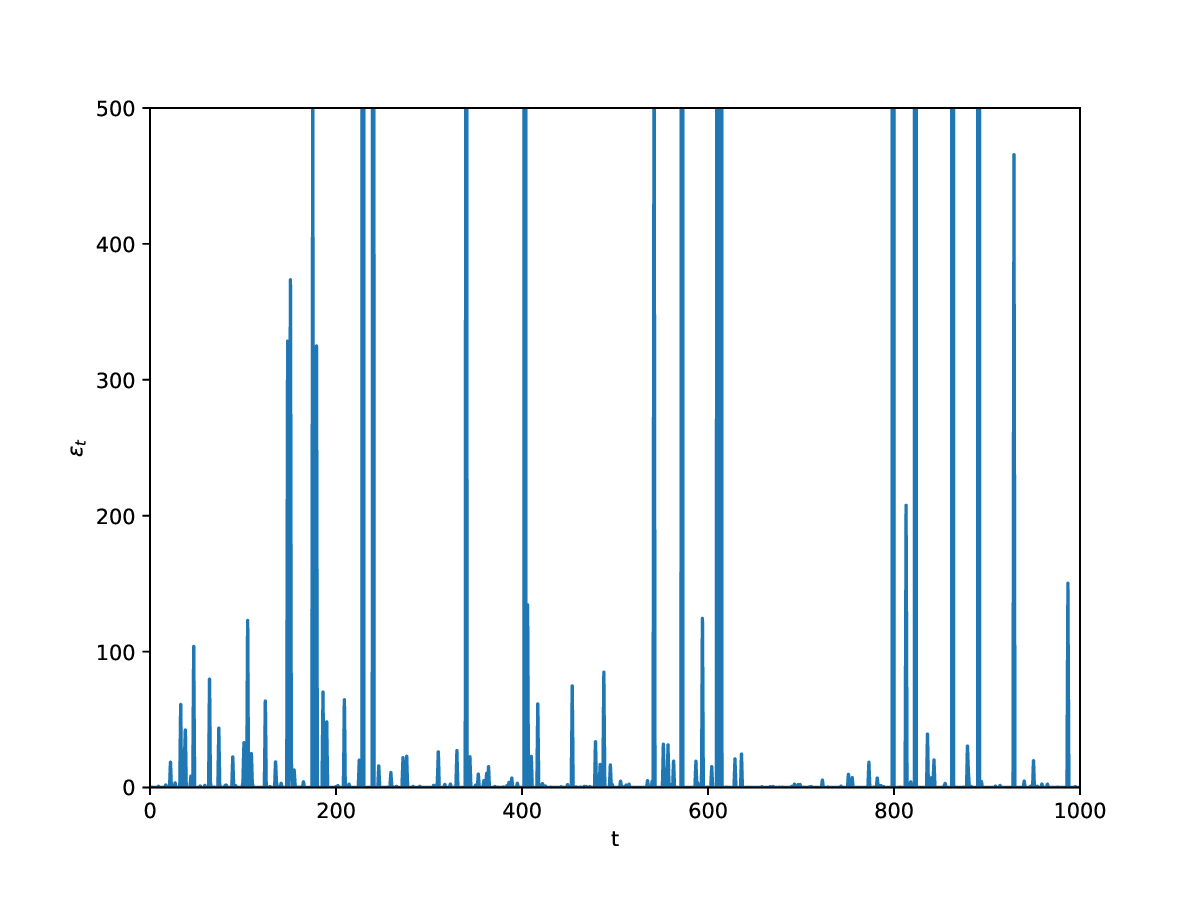}} \hfill\\
	\subfigure[$Y_t$ for $\alpha_2 = 0.6, \rho_2 = 0.8.$]{\includegraphics[width=0.49\textwidth]{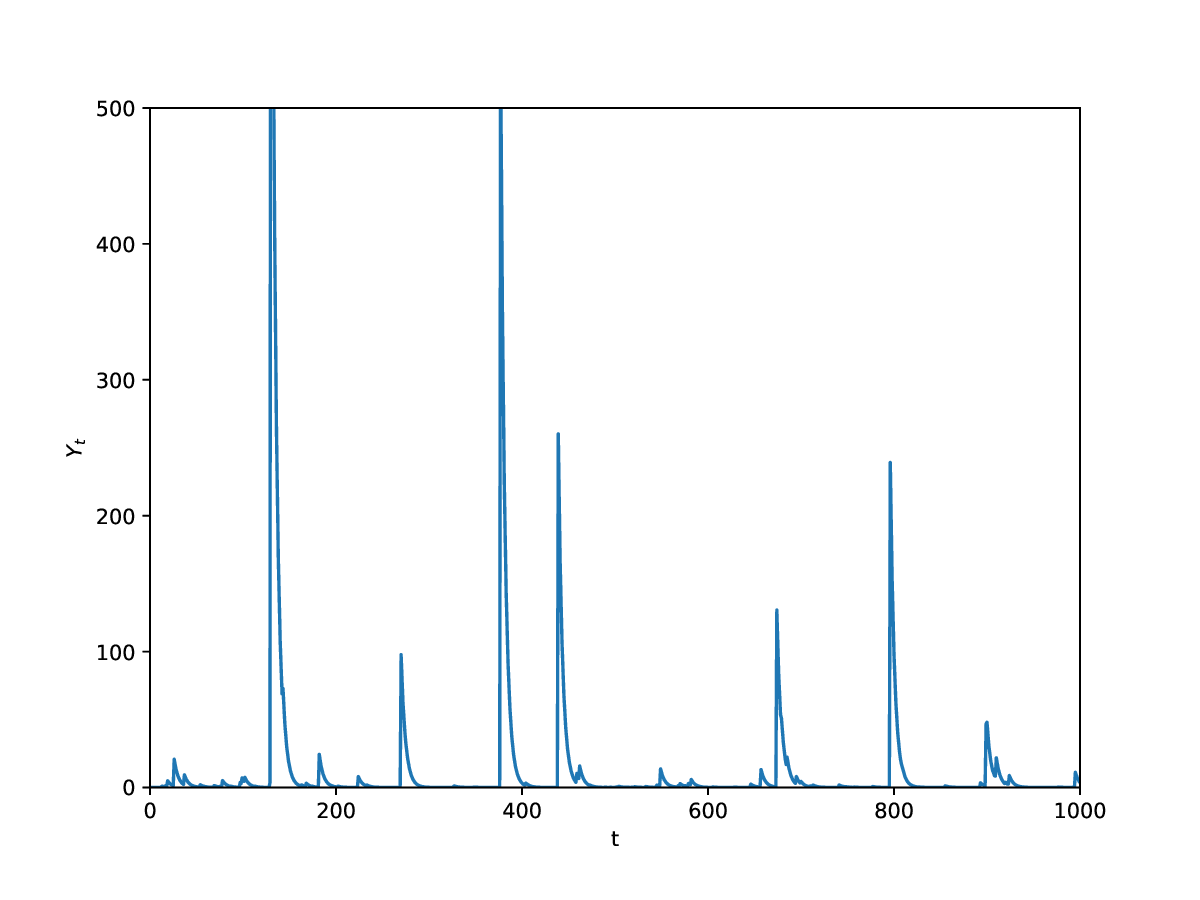}}
	\subfigure[$\epsilon_t$ for ${\alpha_2} = 0.6, \rho_2 = 0.8.$]{\includegraphics[width=0.49\textwidth]{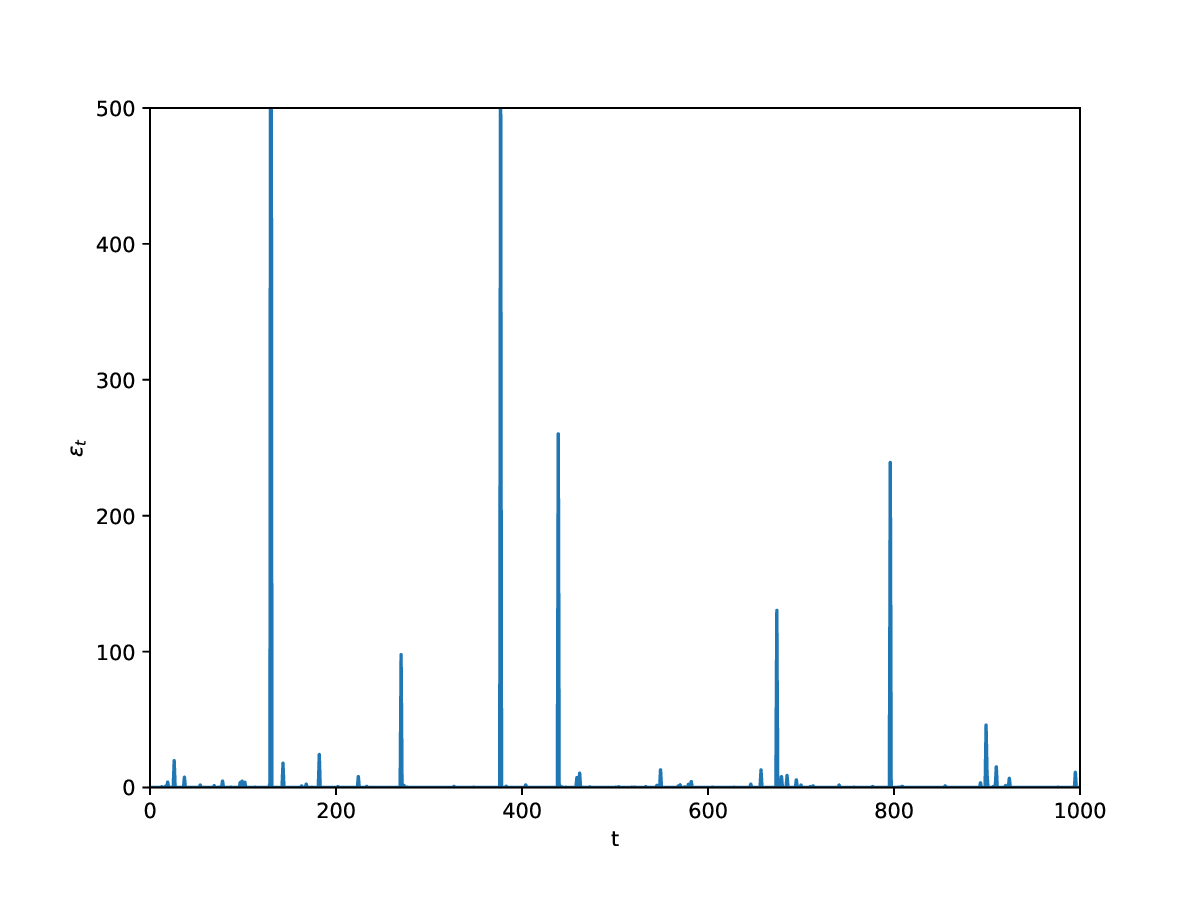}}
	\caption{The plots of time series $\{Y_t\}$ with ML($\alpha,1$) marginals with true parameters (a) $\alpha_1 = 0.4,\, \rho_1=0.4$ (c) $\alpha_2 = 0.6,\, \rho_2 = 0.8$ and the corresponding innovation series $\{\epsilon_t\}$ with true parameters (b) $\alpha_1 = 0.4, \rho_1 = 0.4$ (d) $\alpha_2 = 0.6, \rho_2 = 0.8$.}\label{fig3}
\end{figure}

\begin{table}
	\centering
	\caption{Parameter estimates of data 1 and data 2. }\label{tab3}
	
	\begin{tabular}{lllll}
		\hline\noalign{\smallskip}
		& $\hat{\alpha_1}(\alpha_1 = 0.4) $ & $\hat{\rho_1}(\rho_1 = 0.4)$ & $\hat{\alpha_2}(\alpha_2 = 0.6)$ & $\hat{\rho_2}(\rho_2 = 0.8)$ \\
		\noalign{\smallskip}\hline\noalign{\smallskip}
		& $0.4038$ & $0.4049$ & $0.6003$ & $0.7976$ \\ 
		
		\noalign{\smallskip}\hline
	\end{tabular}

\end{table}
\begin{figure}[H]
	\centering
	\subfigure[Boxplot for $\alpha_1 = 0.4,\rho_1 = 0.4.$]{\includegraphics[width=0.49\textwidth]{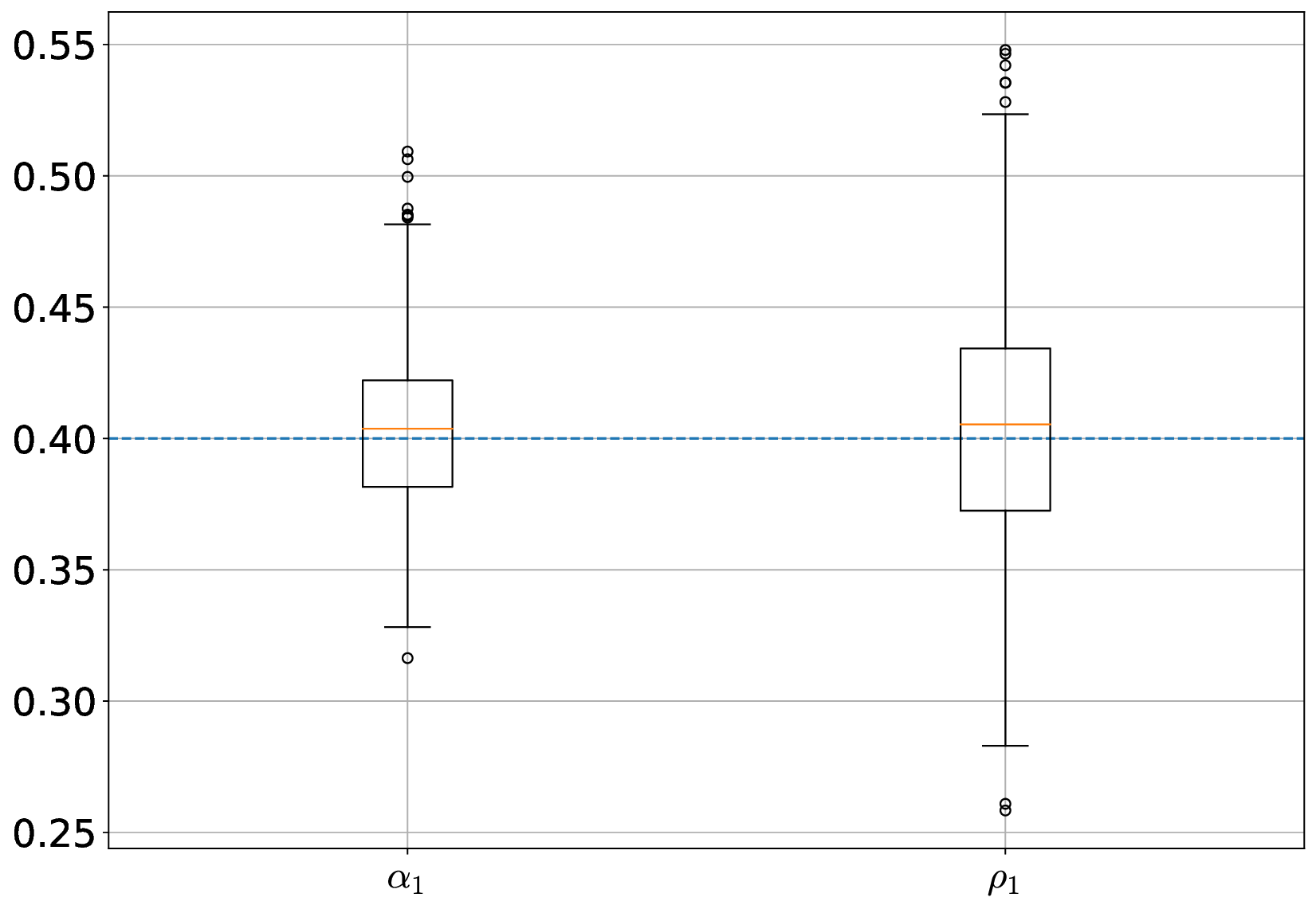}}\hfill
	\subfigure[Boxplot for ${\alpha_2} = 0.6, \rho_2 = 0.8.$]{\includegraphics[width=0.49\textwidth]{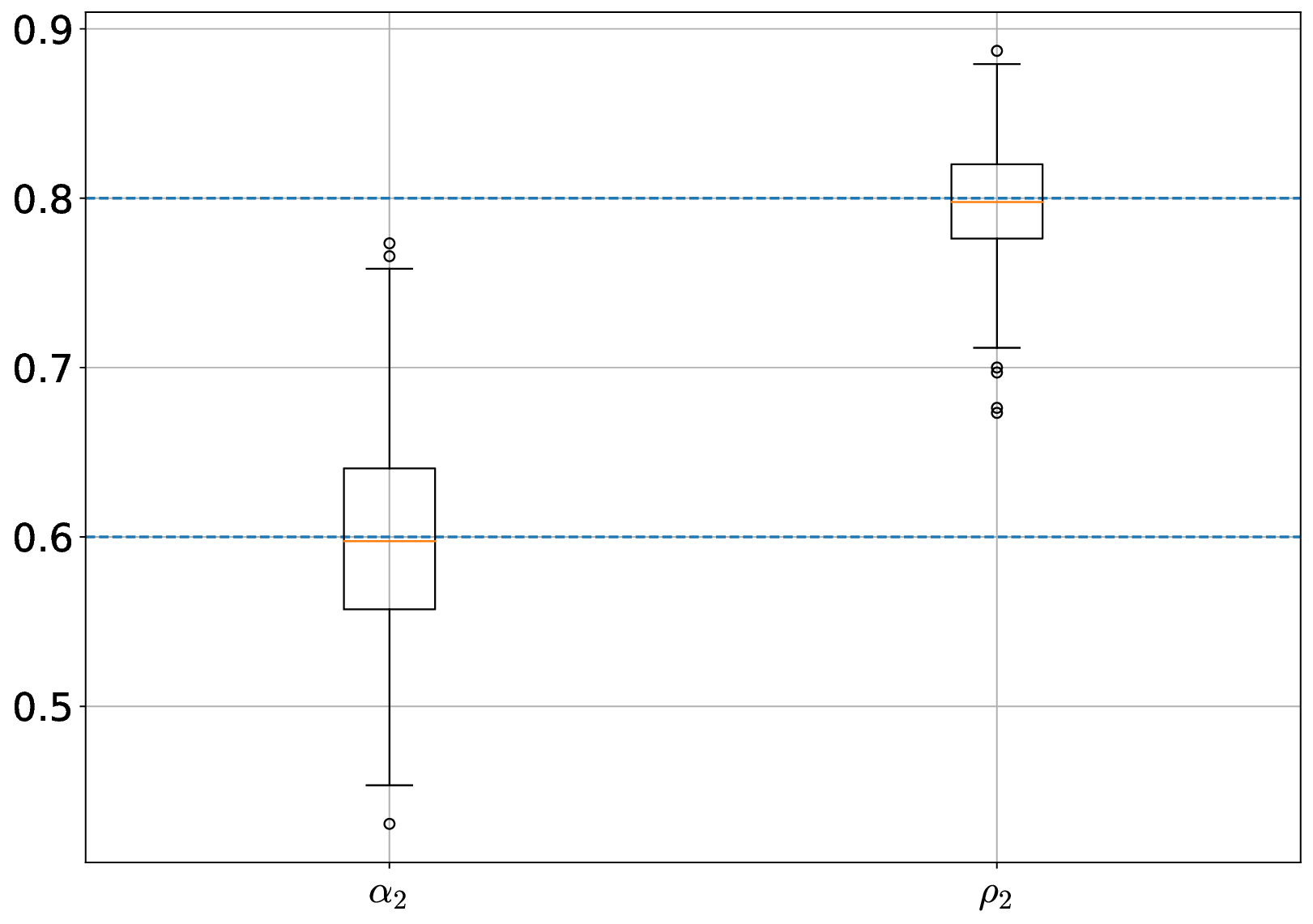}} 
	\caption{The boxplots for the parameter estimates of AR$(1)$ model with ML($\alpha,1$) marginals with true parameter values (a) $\alpha_1 = 0.4,\, \rho_1=0.4$ (b) $\alpha_2 = 0.6,\, \rho_2 = 0.8$. The simulated data has $500$ trajectories each of length $1000.$}\label{fig4}
\end{figure}

\begin{table}
	\centering
	\caption{RMSE and MAE of estimated parameters for data 1 and data 2. }\label{tab4}
	
	\begin{tabular}{lllll}
		\hline\noalign{\smallskip}
		& $\alpha_1 $ & $\rho_1$ & $\alpha_2$ & $\rho_2$ \\
		\noalign{\smallskip}\hline\noalign{\smallskip}
		RMSE & $0.0319$ & $0.0492$ & $0.0600$ & $0.0331$ \\ 
		MAE &  $0.0250$ & $0.0386$ & $0.0478$ & $0.0261$\\
		\noalign{\smallskip}\hline
	\end{tabular}
\end{table}


\section{Conclusion}\label{sec:ML5}
We first discuss some of the well-known properties of ML$(\alpha, \sigma)$ distribution and make use of its mixture representation to estimate the parameters. We propose to use the empirical Laplace transform method to estimate parameters of ML($\alpha, \sigma$) and assess the method's performance on four different set of simulated data. Afterwards, we apply the ML($\alpha, \sigma$) on high frequency trading data for oil futures. We compare the inter-arrival time of positive and negative jumps with exponential distribution. 
Further, the estimation of three-parameter ML($\alpha, \sigma, \gamma$) distribution is also performed using empirical Laplace transform method. The estimation of generalized version of Mittag-Leffler distribution is helpful in modeling various complex stochastic processes. We also provide the AR($1$) process with marginals from ML$(\alpha, 1)$ distribution and deduce the distribution of innovation terms and vice versa. Following this, we propose to use empirical Laplace transform to estimate the parameters of AR($1$) process. The simulation results imply that we can rely on this method for estimation and further model can be generalized.
The Prabhakar distribution provides a flexible framework for modeling heavy-tailed and complex stochastic behavior. As a natural extension of present work, in future we may focus on incorporating the Prabhakar distribution into autoregressive time series frameworks, particularly AR$(p)$ models, where the marginal or innovation structure follows the proposed distribution.

\section*{Acknowledgments}
I would like to express my gratitute to Dr. Arun Kumar for generously sharing the data and his guidance throughout this study. I would also like to thank Ministry of Education, India for supporting the research work. 
\bibliographystyle{plainnat}
\bibliography{bibfile}

\end{document}